\documentclass[journal]{IEEEtran}
\usepackage{amsmath, amsthm, amssymb, graphicx, subfigure, relsize, paralist}

\newtheorem{theorem}{Theorem}
\newtheorem{lemma}[theorem]{Lemma}
\newtheorem{proposition}[theorem]{Proposition}

\newtheorem{example}{Example}  
\newtheorem{remark}{Remark}  

\newcommand{ \sml }[1]{ \mathsmaller{#1} }
\newcommand{ \shift }{ \operatorname{ShC} }
\newcommand{ \nshift }{ \operatorname{NShC} }
\newcommand{ \ctshift }{ \operatorname{CTShC} }
\newcommand{ \queue }{ \operatorname{DTQP} }
\newcommand{ \nqueue }{ \operatorname{NDTQP} }
\newcommand{ \ctqueue }{ \operatorname{CTQP} }

\newcommand{ \argmax }{ \operatorname{argmax} }

\newcommand{ \myqed }{ \hfill $\blacktriangle$ }

\begin{document}

\title{Zero-Error Capacity of\\$ P $-ary Shift Channels and FIFO Queues}

\author{
        Mladen~Kova\v{c}evi\'c,~\IEEEmembership{Member,~IEEE},
				Milo\v{s}~Stojakovi\'{c},  and
				Vincent~Y.~F.~Tan,~\IEEEmembership{Senior~Member,~IEEE}%
\thanks{Date: September 10, 2017.

        M. Kova\v{c}evi\'{c} and V. Y. F. Tan are supported by the Singapore Ministry
        of Education (MoE), under Tier 2 Grant R-263-000-B61-112.
        M. Stojakovi\'{c} is partly supported by the Ministry of Education, Science, and
				Technological Development, Republic of Serbia, and the Provincial Secretariat for
        Higher Education and Scientific Research, Province of Vojvodina.}%
\thanks{M. Kova\v{c}evi\'c is with the Department of Electrical \& Computer Engineering,
        National University of Singapore, Singapore 117583
				(email: mladen.kovacevic@nus.edu.sg).
				
				M. Stojakovi\'{c} is with the Department of Mathematics and Informatics,
        Faculty of Sciences, University of Novi Sad, 21000 Novi Sad, Serbia
        (email: milos.stojakovic@dmi.uns.ac.rs).
				
				V. Y. F. Tan is with the Department of Electrical \& Computer Engineering and the
				Department of Mathematics, National University of Singapore, Singapore 117583
        (email: vtan@nus.edu.sg).}%
}


\maketitle

\begin{abstract}
The objects of study of this paper are communication channels in which the dominant
type of noise are symbol shifts, the main motivating examples being timing and
bit-shift channels.
Two channel models are introduced and their zero-error capacities and zero-error-detection
capacities determined by explicit constructions of optimal codes.
Model A can be informally described as follows:
\begin{inparaenum}[1)]
\item
The information is stored in an $ \boldsymbol n $-cell register, where each cell
is either empty or contains a particle of one of $ \boldsymbol P $ possible types,
and
\item
due to the imperfections of the device each of the particles may be shifted
several cells away from its original position over time.
\end{inparaenum}
Model B is an abstraction of a single-server queue:
\begin{inparaenum}[1)]
\item
The transmitter sends packets from a $ \boldsymbol P $-ary alphabet through a
queuing system with an infinite buffer and a First-In-First-Out (FIFO) service
procedure, and
\item
each packet is being processed by the server for a random number of time slots.
\end{inparaenum}
More general models including additional types of noise that the particles/packets
can experience are also studied, as are the continuous-time versions of these
problems.
\end{abstract}%


\section{Introduction and preliminaries}

In several communication and information storage systems the dominant type
of ``noise'' introduced by the channel are \emph{shifts} of symbols of the
transmitted sequence.
A classic example is the so-called bit-shift or peak-shift channel which
has been introduced as a model for some magnetic recording devices wherein
the electric charges (the $ 1 $-bits) can be shifted to the left or to the
right of their original position due to various physical effects (see, e.g.,
\cite{shamai}).
Another familiar scenario is the transmission of information packets through
a queue with random processing times.
Such a queue is intended to model, e.g., a network router processing the
packets and then forwarding them towards their destination.
The capacity of such channels can in general be increased by encoding the
information in the transmission times of packets in addition to their contents
\cite{anantharam}, in which case the unknown delays of packets at the output
of the queue represent the noise.
Another setting where timing channels naturally arise are molecular
communications \cite{farsad2, nakano}.
The information here is encoded in the number and the types of the particles
released at given time instants, and the noise are random delays that particles
experience on their way to the receiving side, caused by their interaction
with the fluid medium.

Motivated by the above examples, we analyze here two channel models that are
intended to capture such impairments.
In the remainder of this section we shall define these models formally and
describe their relation to the models previously studied in the literature.
In Section \ref{sec:shift}, a construction of optimal zero-error codes for
Model A is given and a characterization of its zero-error capacity is obtained.
In Section \ref{sec:queue} the corresponding results for Model B are derived.
In Section \ref{sec:detect} we determine the zero-error-detection capacity of
the two channels.
Section \ref{sec:ct} contains the analysis of the continuous-time versions
of both models.
A brief conclusion and several pointers for further work are stated in Section
\ref{sec:concl}.

\subsection{Model A}
\label{sec:defshift}

Let $ n, P, K_1, K_2 $ be integers, with $ n, P, K_2 \geq 0 $ and $ K_1 \leq K_2 $.
The channel inputs are sequences of length $ n $ over the alphabet $ \{0, 1, \ldots, P\} $.
Think of such an input sequence $ {\bf x} = (x_\sml{1}, \ldots, x_n) $ as representing
a state of an $ n $-cell register, where $ x_i = 0 $ means that the $ i $'th
cell is empty, while $ x_i = p $, $ p \in \{1, \ldots, P\} $, means that the
$ i $'th cell contains a particle of ``type'' $ p $.
For any such input sequence the channel outputs one of the sequences
$ {\bf z} = (z_{\sml{1+K_1}}, \ldots, z_{n+\sml{K_2}}) $ satisfying the
following conditions:
\begin{inparaenum}[1)]
\item
$ \bf z $ is of length $ n' = n + K_2 - K_1 $,
\item
The subsequences $ \tilde{\bf x} = (x_{i_1}, \ldots, x_{i_m}) $ and
$ \tilde{\bf z} = (z_{j_1}, \ldots, z_{j_{m'}}) $, obtained by deleting all
the zeros in $ \bf x $ and $ \bf z $ respectively, are identical (and hence
$ m = m' $), and
\item
$ K_1 \leq j_l - i_l \leq K_2 $ for all $ 1 \leq l \leq m $.
\end{inparaenum}
Each of these sequences is output with positive probability.
If $ \bf x $ can produce $ \bf z $ at the channel output, we write
$ {\bf x} \rightsquigarrow {\bf z} $.	

In words, the $ i $'th particle is shifted $ k_i $ cells to the right of its
original position over time, where $ K_1 \leq k_i \leq K_2 $, but no two particles
can swap cells or end up in the same cell (if $ k_i < 0 $, then this is of course a
shift to the left for $ |k_i| $ cells).
We assume that there are enough empty cells, to the left or to the right of the
register, for the boundary particles to be able to shift; this assumption simplifies
the analysis slightly but has no influence on the results.

The channel just described will be referred to as the $ P $-ary Shift Channel
with parameters $ K_1, K_2 $, or $ \shift(P ; K_1, K_2) $ for short.
$ \shift(P; K) $ will stand for $ \shift(P; 0, K) $.
Further generalization of this model including additional types of noise will
be discussed in Section \ref{sec:shiftaddnoise}, and its continuous-time version
in Section \ref{sec:ctshift}.

\subsection{Model B}
\label{sec:defqueue}

Let $ n, P, K $ be nonnegative integers.
The channel inputs are sequences of length $ n $ over the alphabet $ \{0, 1, \ldots, P\} $,
but we now think of a sequence $ {\bf x} = (x_\sml{1}, \ldots, x_n) $ as describing
a stream of packets entering a queue, $ x_i = 0 $ meaning that the $ i $'th time
slot is empty, and $ x_i = p $, $ p \in \{1, \ldots, P\} $, that a packet of ``type''
$ p $ was transmitted in that slot.
For any such input sequence the channel outputs one of the sequences
$ {\bf z} = (z_{\sml{1}}, \ldots, z_{n'}) $, satisfying the following conditions:\linebreak
\begin{inparaenum}[1)]
\item
$ \bf z $ is of length $ n' \geq n $, and if $ n' > n $ its last symbol, $ z_{n'} $,
is nonzero,
\item
The subsequences $ \tilde{\bf x} = (x_{i_1}, \ldots, x_{i_m}) $ and
$ \tilde{\bf z} = (z_{j_1}, \ldots, z_{j_{m'}}) $, obtained by deleting all
the zeros in $ \bf x $ and $ \bf z $, respectively, are identical (and hence
$ m = m' $), and
\item
$ 0 \leq j_l - \max\{ i_l, j_{l-1} + 1 \} \leq K $ for all $ 1 \leq l \leq m $,
where $ j_0 = 0 $.
\end{inparaenum}
Each of these sequences is output with positive probability.

In words, the first packet is delayed for at most $ K $ slots due to processing
(it was sent in slot $ i_1 $ and received in slot $ j_1 $).
If the second packet arrives at the queue while the first packet is being processed,
it has to wait for the server to become free, and the first available slot when it
itself starts being processed is $ j_1 + 1 $; otherwise it can be processed immediately
when it arrives, which is in slot $ i_2 $, etc.
Thus, every packet waits in the queue for the server to become free---so-called
First-In-First-Out (FIFO) service procedure---and is then processed for a randomly
chosen number of slots, this number being $ \leq K $.
Observe that the total delay of a packet can be much larger than $ K $ due to
the possibility of waiting in the queue, and consequently the output sequence
can be as long as $ (K + 1) n $.

As we shall explain shortly (see Section \ref{sec:ze}), the probabilistic
description of this channel needs to be specified too, even though we are
analyzing only zero-error problems.
We assume that each packet is processed for $ k \in \{ 0, 1, \ldots, K \} $
slots with probability $ \varphi(k) > 0 $, where $ \sum_{k=0}^K \varphi(k) = 1 $,
independently of everything else.
Denoting the random variable which represents the processing time by $ \kappa $,
the average processing time of a packet can be written as
$ \mathbb{E}_\varphi[\kappa] = \sum_{k=0}^K k \varphi(k) $.

The channel described above will be referred to as the Discrete-Time Queue with
bounded Processing time, $ \queue(P; K; \varphi) $.
Its generalization including additional types of noise will also be discussed in Section
\ref{sec:queueaddnoise}, and its continuous-time version in Section \ref{sec:ctqueue}.

Note that the shift channel $ \shift(P; K) $ can also be seen as a discrete-time
queue with an infinite buffer and a FIFO service procedure, but in which the
\emph{residence} times of the packets are bounded by $ K $, rather than their
processing times (the residence time is the total time the packet spends in the
system, either waiting to be processed, or being processed).

\subsection{Zero-Error Codes and Zero-Error Capacity}
\label{sec:ze}

An error-correcting code of length $ n $ for a particular channel is a nonempty
subset of the set of all possible inputs of length $ n $.
A code $ {\mathcal C}(n) $ is said to be a \emph{zero-error} code if its error
probability is equal to zero under optimal decoding.
In other words, the requirement is that all possible errors allowed in the model
can be corrected, or equivalently, that no two different codewords
$ {\bf x}, {\bf y} \in {\mathcal C}(n) $ can produce the same sequence $ \bf z $
at the channel output.

For a given code $ {\mathcal C}(n) $, denote by $ L_{\text{av}}(n) $ the average
length of the channel output, the average being taken over all codewords and
channel statistics.
(The dependence of $ L_{\text{av}}(n) $ on the code and the channel is suppressed
for notational simplicity.)
In symbols,
$ L_{\text{av}}(n) =
      \frac{ 1 }{ | {\mathcal C}(n) | }
	       \sum_{ {\bf x} \in {\mathcal C}(n) }   \sum_{ {\bf z} }
            |{\bf z}| \cdot \text{Pr}\{ {\bf x} \rightsquigarrow {\bf z} \} $,
where $ |{\bf z}| $ denotes the length of a sequence $ \bf z $ and
$ \text{Pr}\{ {\bf x} \rightsquigarrow {\bf z} \} $ the probability that $ \bf z $
is obtained at the channel output when $ \bf x $ is at its input.

\begin{example}
\label{exmplL}
\textnormal{
Consider a code $ {\mathcal C}(n) $ for the $ \queue(1; K; \varphi) $ consisting
of a single codeword $ {\bf x} = (1, \ldots, 1) $ ($ n $ identical packets sent in
$ n $ successive slots).
Denoting the processing time of the $ i $'th packet by $ \kappa_i $, we can
express the length of the output sequence $ \bf z $ as
$ L(n) = \sum_{i=1}^n (1 + \kappa_i) $
($ \kappa_i $'s are assumed independent and distributed according to $ \varphi $).
Its average value is
$ L_{\text{av}}(n) = n (1 + \mathbb{E}_\varphi[\kappa_1])$.
This fact will be used in the proof of Theorem \ref{thm:queue}.
}
\myqed
\end{example}

We define the rate of a code $ {\mathcal C}(n) $ as
$ \frac{1}{L_{\text{av}}(n)} \log |{\mathcal C}(n)| $,
where $ \log $ is to the base $ 2 $.
Finally, the zero-error capacity of a channel is the $ \limsup $ of the rates of
optimal zero-error codes (i.e., zero-error codes having the largest possible
cardinality) of length $ n \to \infty $ for that channel.

\vfill
\begin{remark}[Code rate]
\label{rem:rate}
\textnormal{
The above definition of the code rate may seem a bit unusual so we shall elaborate.
In channels with shifts and delays, the length of the output sequence is a random
variable and is in general different from the length of the corresponding input.
Therefore, normalizing the number of transmitted bits of information, $ \log |{\mathcal C}(n)| $,
by the average time it takes the receiver to obtain the entire sequence, $ L_{\text{av}}(n) $,
is a natural measure of rate of transmission through such channels.
In channels where the length of each possible output is the same as the length of
the corresponding input, we have $ L_{\text{av}}(n) = n $ and the definition of rate
reduces to the usual one.
More generally, when $ L_{\text{av}}(n) = n + o(n) $, we can again use the standard
definition for the purpose of determining the capacity because only the asymptotic
behavior is relevant here.
This is the case in the $ \shift(P; K) $ for instance, where $ L_{\text{av}}(n) \leq n + K $.
However, in the case of the $ \queue $ the length of the output can differ
from that of the corresponding input by a multiplicative constant, and the
actual behavior of $ L_{\text{av}}(n) $ will have to be taken into account.
This is the reason why the probability distribution $ \varphi $ is included
in the description of the $ \queue $---the zero-error capacity of this channel
in general depends on it, or at least on its mean.
}
\myqed
\end{remark}

\vfill
\begin{remark}[Concatenated codewords]
\label{rem:zp}
\textnormal{
  If one is interested in the regime of communication where multiple codewords
are being sent in succession, then the notion of zero-error code needs to be
redefined because shifts of symbols can cause interference between successive
codewords.
The requirement in that case is that no two \emph{sequences of codewords}
can produce the same output \cite[Def.\ 2]{kovacevic+popovski}.
The zero-error capacity, however, is the same under both definitions.
}
\myqed
\end{remark}

\begin{remark}[Zero-error capacity]
\label{rem:cap}
\textnormal{
  Intuitively, the zero-error capacity of a channel should be defined as
the \emph{supremum} of the rates of all zero-error codes for that channel.
In most of the studied models this supremum is equal to the $ \limsup $, and
in fact to the \emph{limit} of the rates of optimal codes \cite{korner}.
This does not necessarily hold for the channels treated here---a zero-error
code of length $ n $ may have rate higher than the capacity.
This is a consequence of the definition of the code rate via $ L_{\text{av}}(n) $,
and especially manifests itself in the case of the $ \queue $.
It should be noted, however, that only a bounded amount of information, i.e.,
a fixed number of bits, can be transmitted at such a rate because the code
is of finite length, and sending multiple codewords in succession does not
guarantee that the zero-error property will be preserved (Remark \ref{rem:zp}).
Adopting the $ \limsup $ definition seems to be necessary in order to determine
the zero-error capacity analytically, and this quantity then has the meaning of
the largest rate at which an unbounded amount of information can be transmitted
error-free.
}
\myqed
\end{remark}

\vspace{-1mm}
\subsection{Previous Work}

Models most closely related to the shift channel introduced in Section \ref{sec:defshift}
are those in \cite{shamai, krachkovsky, kovacevic+popovski}.
In particular, \cite{krachkovsky} studies the zero-error capacity of the bit-shift
channel $ \shift(1; -K, K) $ under additional constraints on input sequences (the
so-called $ (d,k) $-runlength limited sequences \cite{immink}), and \cite{kovacevic+popovski}
studies a generalization of the $ \shift(1; K) $ wherein multiple (but identical)
particles per slot are allowed.
We analyze here generalizations of these models that include arbitrary shifts
($ K_1, K_2 $), multiple types of particles ($ P \geq 1 $), additional types of
noise that these particles can experience, and the continuous-time models.
We also mention the work \cite{engelberg+keren} where a particular kind of
shift channel was studied and bounds on its zero-error capacity derived.
The exact value of the zero-error capacity for that model was determined in
\cite{kovacevic} using methods very similar to those used here.

The zero-error-detection problem that we address in Section \ref{sec:detect}
has not been studied before for shift channels, timing channels, and the like.

As for queuing channels such as the $ \queue $, this is to our knowledge the
first work addressing zero-error problems for such models.
\emph{Shannon} capacity of queuing systems, on the other hand, is relatively well-studied.
The seminal work on this subject is \cite{anantharam} (continuous-time case), which
was followed by \cite{bedekar, thomas} (discrete-time case);
models with bounded processing time were analyzed in \cite{sellke}.
Our work may be seen as the zero-error counterpart of these and similar
information-theoretic studies of queuing systems.

\section{Zero-error capacity of the shift channel}
\label{sec:shift}

In this section we study error-free communication through the shift channel and
give a characterization of its zero-error capacity.
We also state a generalization of these results to the case where the channel
introduces some other types of noise in addition to the shifts.

\subsection{Reduction to the $ \shift(1; K) $}
\label{sec:reduction}

Before proceeding with the analysis, we point out in this subsection several
simple, but important facts about the effect of the $ \shift(P; K_1, K_2) $
on the input sequences.
The first such observation is that codes for this channel depend only on
$ K = K_2 - K_1 $ and not on the particular values $ K_1, K_2 $, which means
that there is no loss in generality in focusing on the case
$ \shift(P ; K) \equiv \shift(P ; 0, K) $.

\begin{lemma}
\label{thm:noleft}
  Every zero-error code for the $ \shift(P ; K_1, K_2) $ is a zero-error code
for the $ \shift(P ; K_2-K_1) $, and vice versa.
\end{lemma}
\begin{proof}
  Just observe that the receiver can shift \emph{all} the received particles
for another $ K_1 $ cells to the left (or, alternatively, shift its point of
reference $ K_1 $ cells to the right) and thus ``create" the channel with
parameters $ 0 $ and $ K_2-K_1 $.
This clearly does not affect the decoding process and the zero-error property
of the code.
\end{proof}

The second observation is that the shift-channel does not affect the Hamming
weight of the transmitted codeword.
This implies that an optimal zero-error code of length $ n $ for the $ \shift(P; K) $
is the disjoint union of optimal zero-error codes of length $ n $ and weight $ W $,
over all $ W \in \{ 0, 1, \ldots, n \} $.
Denoting the cardinality of these codes by $ M_\sml{P;K}(n) $ and $ M_\sml{P;K}(n,W) $,
respectively, we can write
\begin{equation}
\label{eq:nvianW}
  M_\sml{P;K}(n) = \sum_{W=0}^n M_\sml{P;K}(n,W) .
\end{equation}
Therefore, it suffices to focus on the constant-weight case.

Finally, the analysis of communication with several types of
particles can be reduced to that with a single type only, i.e., $ P = 1 $.
In other words, we can treat the information contained in the positions of
the particles and that in the types of the particles separately (see also
\cite[Sec.\ IV]{anantharam}).
Before stating this more formally, we introduce two notational conventions:
For $ {\bf x} \in \{0, 1, \ldots, P\}^n $, let $ \underline{\bf x} $ denote
its indicator sequence---binary sequence having zeros at the same positions
as $ \bf x $, i.e., $ x_i \neq 0  \Leftrightarrow  \underline{x}_i = 1 $, and
let $ \tilde{\bf x} $ be the sequence obtained by deleting all the zeros in
$ \bf x $.

\begin{proposition}
\label{thm:CPK}
  Let $ \mathcal{C}_\sml{1;K}(n) $ be an optimal zero-error code of length
$ n $ for the $ \shift(1; K) $.
Then
\begin{equation}
\label{eq:CPK}
  \mathcal{C}_\sml{P;K}(n) =
	  \Big\{ {\bf x} \in \{0, 1, \ldots, P\}^n \; : \; \underline{\bf x} \in \mathcal{C}_\sml{1;K}(n) \Big\}
\end{equation}
is an optimal zero-error code of length $ n $ for the $ \shift(P; K) $.
\end{proposition}
\begin{IEEEproof}
Since insertions, deletions and reordering of particles are not possible, two
sequences $ {\bf x}, {\bf y} \in \{0, 1, \ldots, P\}^n $ can be confusable (i.e.,
can produce the same output) in the $ \shift(P; K) $ only if the subsequences
$ \tilde{\bf x} $ and $ \tilde{\bf y} $, obtained by deleting the zeros in $ \bf x $
and $ \bf y $ respectively, are identical.
Furthermore, sequences $ \bf x, \bf y $ with $ \tilde{\bf x} = \tilde{\bf y} $,
are confusable in the $ \shift(P; K) $ if and only if $ \underline{\bf x} $ and
$ \underline{\bf y} $ are confusable in the $ \shift(1; K) $.
This implies that the code $ \mathcal{C}_\sml{P;K}(n) $, as defined in \eqref{eq:CPK},
is zero-error.
It also implies that $ \mathcal{C}_\sml{P;K}(n) $ is optimal because a zero-error
code for the $ \shift(P;K) $ can have at most $ M_\sml{1;K}(n,W) $ codewords
$ {\bf x} $ having the same subsequence $ \tilde{\bf x} $ of Hamming weight $ W $,
and so $ M_\sml{P;K}(n,W) \leq P^{W} M_\sml{1;K}(n,W) $ and
$ M_\sml{P;K}(n) \leq \sum_{W=0}^n P^{W} M_\sml{1;K}(n,W) = | \mathcal{C}_\sml{P;K}(n) | $.
\end{IEEEproof}

\subsection{Optimal Codes and the Capacity}
\label{sec:shift_opt}

As demonstrated above, one can focus first on the special case of $ \shift(1;K) $
and obtain the results for the general case by using Lemma \ref{thm:noleft} and
Proposition \ref{thm:CPK}.
Optimal codes for this channel have in fact been determined in \cite{kovacevic+popovski},
but we shall rederive here this result and give an alternative proof of optimality
by focusing on the constant-weight case.
This approach will lead to an even simpler---geometric---characterization of optimal
codes, and will enable a unified treatment of many related problems, such as the
$ \queue $ channel, the continuous-time models, the error-detection problem, etc.

Let us describe the set of constant-weight inputs to the $ \shift(1;K) $ in a
way appropriate for our purpose.
Binary sequences of length $ n $ and weight $ W $ can be uniquely represented
as $ W $-tuples of positive integers $ (s_1, \ldots, s_\sml{W}) $, where $ s_i $
is the position of the $ i $'th $ 1 $-bit in the sequence;
for example, $ 010001 \leftrightarrow (2,6) $.
The set of all such sequences is therefore in a one-to-one correspondence with the simplex
$ \big\{ (s_1, \ldots, s_\sml{W}) \in \mathbb{Z}^{W} : 1 \leq s_1 < \cdots < s_\sml{W} \leq n \big\} $.
For notational convenience, we shall subtract the vector $ (1, \ldots, W) $
from all vectors in this set to obtain another equivalent representation,
$ \Delta_{n\sml{-W}}^\sml{W} =
  \big\{ (s_1, \ldots, s_\sml{W}) \in \mathbb{Z}^{W} : 0 \leq s_1 \leq \cdots \leq s_\sml{W} \leq n-W \big\} $.
According to our channel model, the set of outputs $ \{ {\bf z}:  {\bf x} \rightsquigarrow {\bf z} \} $
is in this representation the hypercube of sidelength $ K $ with $ \bf x $ at its
corner (restricted to the simplex), namely
$ \big\{ (z_1, \ldots, z_\sml{W}) \in \mathbb{Z}^{W} :
0 \leq z_1 \leq \cdots \leq z_\sml{W} \leq n - W + K,\linebreak
0 \leq z_i - x_i \leq K \big\} $
(the $ 1 $-bits of the transmitted binary sequence $ \bf x $ are shifted to the
right for $ \leq K $ positions in the channel).
Figure \ref{fig:zecode} depicts the just described representation of the set
of binary sequences of length $ n = 9 $ and weight $ W = 2 $, as well as the
effect of the $ \shift(1; 1) $ on these sequences.	
We generally do not distinguish between binary sequences and their integer
representations; it will be clear from the context which description is used.

\begin{figure}[h]
\centering
  \includegraphics[width=0.82\columnwidth]{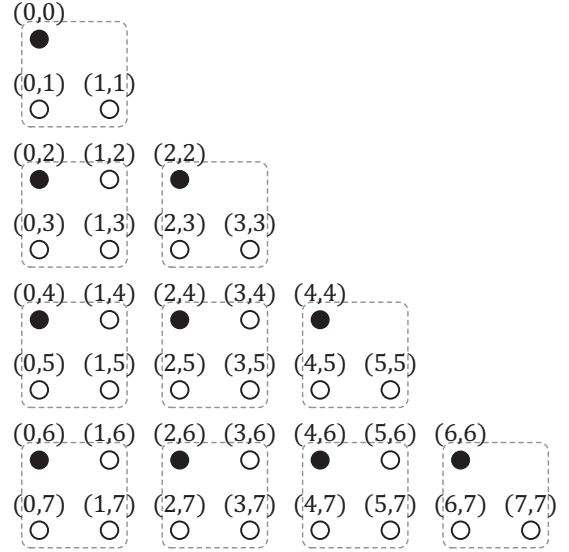}
\caption{The simplex $ \Delta_{7}^2 $ representing the set of binary sequences
         of length $ n = 9 $ and weight $ W = 2 $; the point $ (i, j) $ represents the
         binary sequence having $ 1 $'s on the $ (i+1) $'th and $ (j+2) $'th position.
				 Black dots denote the codewords of the code $ {\mathcal C}_\sml{1;1}(9,2) $---an
         optimal zero-error code for the $ \shift(1;1) $.
				 Dashed lines illustrate sets of sequences that a given codeword can
         produce at the output of this channel.}
\label{fig:zecode}
\end{figure}%

\begin{theorem}
\label{thm:simplexcode}
The code
\begin{equation}
\label{eq:simplexcode}
  {\mathcal C}_\sml{1;K}(n,W) = \Big\{ {\bf x} \in \Delta_{n\sml{-W}}^\sml{W} \; : \; {\bf x} = {\bf 0} \;\; (\operatorname{mod} \; K+1)  \Big\}
\end{equation}
is an optimal zero-error code of length $ n $ and weight $ W $ for the $ \shift(1;K) $.
\end{theorem}
\begin{IEEEproof}
  The sets of outputs
$ \{ {\bf z}:  {\bf x} \rightsquigarrow {\bf z} \} $ and
$ \{ {\bf z} : {\bf y} \rightsquigarrow {\bf z} \} $
(hypercubes of sidelength $ K $) are disjoint for every two distinct
codewords $ {\bf x}, {\bf y} \in {\mathcal C}_\sml{1;K}(n,W) $ because, by
construction, the coordinate-wise differences $ x_i - y_i $ are integral multiples
of $ K + 1 $.
This proves that the code $ {\mathcal C}_\sml{1;K}(n,W) $ is zero-error.
Observe also that $ {\mathcal C}_\sml{1;K}(n,W) $ is ``perfect'', in the sense
that the sets of outputs $ \{ {\bf z}:  {\bf x} \rightsquigarrow {\bf z} \} $,
$ {\bf x} \in {\mathcal C}_\sml{1;K}(n,W) $, cover the entire space $ \Delta_{n\sml{-W}}^\sml{W} $.
Indeed, for an arbitrary point $ {\bf z} = (z_1, \ldots, z_\sml{W}) \in \Delta_{n\sml{-W}}^\sml{W} $,
consider the point $ {\bf x} $ defined by $ x_i = \lfloor \frac{z_i}{K + 1} \rfloor (K + 1) $;
then $ {\bf x} \in {\mathcal C}_\sml{1;K}(n,W) $ and
$ {\bf x} \rightsquigarrow {\bf z} $.
It follows from the result of Shannon \cite[Thm 3]{shannon} that such a perfect
code for the $ \shift(1;K) $ is necessarily optimal.
(In the terminology of \cite{shannon}, the mapping
$ \{ {\bf z}:  {\bf x} \rightsquigarrow {\bf z} \}  \mapsto  {\bf x} $, $ {\bf x} \in {\mathcal C}_\sml{1;K}(n,W) $,
is an ``adjacency reducing mapping''.)
\end{IEEEproof}

\pagebreak
Hence, the cardinality of optimal constant-weight codes for this channel is
$ M_\sml{1;K}(n,W) = | {\mathcal C}_\sml{1;K}(n,W) | $.
To determine this quantity explicitly, write $ {\mathcal C}_\sml{1;K}(n,W) $ in
a different form as
$ {\mathcal C}_\sml{1;K}(n,W) = (K+1) \cdot \Delta_{d}^\sml{W} =
\big\{ (K+1)\cdot{\bf s} : {\bf s} \in \Delta_{d}^\sml{W} \big\} $,
where $ d = \big\lfloor \frac{n-W}{K+1} \big\rfloor $.
It follows that
\begin{equation}
\label{eq:M1KW}
  M_\sml{1;K}(n,W) = \big| \Delta_{d}^\sml{W} \big|
                   = \binom{ W + \big\lfloor \frac{n-W}{K+1} \big\rfloor }{ W }
\end{equation}
and, by Proposition \ref{thm:CPK}, we have for arbitrary $ P $
\begin{equation}
\label{eq:MPKW}
  M_\sml{P;K}(n,W) = P^{W} \binom{ W + \big\lfloor \frac{n-W}{K+1} \big\rfloor }{ W } .
\end{equation}

\begin{theorem}
\label{thm:cap}
  The zero-error capacity of the $ \shift(P;K) $ equals $ {\mathcal R}_\sml{P;K}^* = \log r $,
where $ r $ is the unique positive real root of the polynomial $ x^{{K+1}} - P x^{K} - 1 $.
\end{theorem}
\begin{IEEEproof}
  The required capacity is equal to the limit of the rates of optimal codes, so we
only need to determine the asymptotic behavior of $ M_\sml{P;K}(n) $.
In order to do this we write $ M_\sml{P;K}(n) $ in a recurrent form:
\begin{equation}
\label{eq:MPK}
  M_\sml{P;K}(n) = P\cdot M_\sml{P;K}(n-1) + M_\sml{P;K}(n-K-1) ,
\end{equation}
with $ M_\sml{P;K}(n) = 1 + P + \cdots + P^n $ for $ 0 \leq n \leq K $,
which is easily verified from \eqref{eq:nvianW} and \eqref{eq:MPKW}.
Since $ M_\sml{P;K}(n) $ is the solution of the linear recurrence \eqref{eq:MPK},
it can be expressed in terms of the roots of its characteristic polynomial
$ p(x) = x^{{K+1}} - P x^{K} - 1 $ \cite{wilfbook}.
Namely, $ M_\sml{P;K}(n) = \sum_{k=0}^K a_k r_k^n $, where $ r_k $'s are the
roots of $ p(x) $ and $ a_k $'s are complex constants determined by the initial
conditions%
\footnote{Strictly speaking, the expression $ M_\sml{P;K}(n) = \sum_{k=0}^K a_k r_k^n $
is valid only if all the roots are different, so let us verify that they are.
Observe that the unique positive root satisfies $ r > P $ because $ r^K = (r-P)^{-1} $.
Now, if some $ r_j $ had multiplicity two we would have $ p(x) = (x - r_j)^2 q(x) $
and, by calculating the derivatives of both sides, $ (K+1) x^K - P K x^{K-1} = (x - r_j) s(x) $.
This would imply $ r_j = P K / (K+1) < P $, a contradiction.}.
It is known \cite{wilf}, \cite[Ch.\ 3, Thm 2]{wilfbook} that polynomials of this
form (leading coefficient positive, remaining coefficients negative) have a unique
positive real root $ r_0 = r $ and that the remaining roots cannot exceed $ r $ in
modulus, $ |r_k| \leq r $ (in fact, it is easy to show that this inequality is strict
in the case of $ p(x) $).
This implies that $ \lim_{n \to \infty} \frac{1}{n} \log M_\sml{P;K}(n) = \log r $,
as claimed.
\end{IEEEproof}

\subsubsection*{Finite-Length Performance}
It follows from the above proof that $ M_\sml{P;K}(n) \sim a r^n $ (meaning that
$ \lim_{n \to \infty} \frac{ M_\sml{P;K}(n) }{ a r^n } = 1 $), where the constant
$ a $ is determined by the initial conditions of \eqref{eq:MPK}.
We therefore have a finer asymptotic expansion
\begin{equation}
\label{eq:expandM}
  \log M_\sml{P;K}(n) = n \cdot {\mathcal R}_\sml{P;K}^* + \log a + o(1) 
\end{equation}
which indicates not only the limit of the rates of optimal codes (the capacity), but
also the speed of convergence to the limit.

By using Stirling's approximation, we can also find from \eqref{eq:MPKW} the
asymptotics of $ M_\sml{P;K}(n,W) $ when $ n \to \infty $ and $ W \sim wn $,
$ w \in (0,1) $:
\begin{equation}
\label{eq:cwcap}
 \begin{aligned}
  {\mathcal R}_\sml{P;K}(w) 
					&\triangleq \lim_{n\to\infty} \frac{1}{n} \log M_\sml{P;K}(n,wn)  \\
          &= \frac{ w K + 1 }{ K + 1 } {\mathcal H}\left( \frac{ w(K+1) }{ w K + 1 } \right)  + w \log P ,
 \end{aligned}
\end{equation}
where $ {\mathcal H}(\cdot) $ is the binary entropy function.
This quantity can be interpreted as the ``constant-weight zero-error capacity''
of the $ \shift(P;K) $---the largest rate attainable asymptotically with the
requirement that the fraction of the cells containing a particle is (approximately)
$ w $.
Since there are linearly many weights, the zero-error capacity is achievable
with constant-weight codes, and so another way to characterize it is
\begin{equation}
 \begin{aligned}
  {\mathcal R}_\sml{P;K}^* &= \sup_{0 \leq w \leq 1} {\mathcal R}_\sml{P;K}(w) \\
                           &= \frac{w^* K + 1}{K + 1} {\mathcal H}\left( \frac{w^*(K+1)}{w^* K + 1} \right)  + w^* \log P ,
 \end{aligned}
\end{equation}
where $ w^* $ is the maximizer of $ {\mathcal R}_\sml{P;K}(w) $.
From Stirling's approximation we can in fact get more information about the
asymptotics of the rates of optimal constant-weight codes:
\begin{equation}
\label{eq:expandMW}
  \log M_\sml{P;K}(n, w^* n) =
      n \cdot {\mathcal R}_\sml{P;K}^* - \frac{1}{2} \log n + {\mathcal O}(1) .
\end{equation}
The expressions \eqref{eq:expandM} and \eqref{eq:expandMW} are akin to the fundamental
bounds on the finite-length performance of optimal codes with non-vanishing error
probabilities studied in Shannon theory \cite{tan}.
Comparing them we see that, even though the capacity can be achieved with constant-weight
codes, their finite-length performance is worse than that of general codes.
This is quantified by the ``second-order'' term $ -\frac{1}{2} \log n $, which
represents the penalty paid for using constant-weight codes.

Some properties of the capacity and related quantities mentioned in this subsection,
and their behavior as functions of the channel parameters, are stated in the Appendix.

\subsection{Additional Noise}
\label{sec:shiftaddnoise}

In many realistic scenarios the ``particles'', apart from being shifted, suffer
from other impairments as well.
For example, a packet passing through a queuing system may also be received
erroneously or may be erased (meaning that the symbol `E' is received instead),
see \cite[Sec.\ IV]{anantharam}.
Suppose that these additional impairments are modeled by a discrete memoryless
channel with input alphabet $ \{1, \ldots, P\} $, with output alphabet \emph{not}
containing%
\footnote{The symbol $ 0 $ has a meaning in the shift channel---it represents an
empty cell. Therefore, if a symbol $ p \in \{1, \ldots, P\} $ could produce a $ 0 $,
this would correspond to a \emph{deletion} of a particle being possible in the
compound channel, in which case our analysis would not apply.}
the symbol $ 0 $, and with zero-error capacity equal to $ \mathsf{C}_0 $
(this channel acts on the particles independently of their shifts; in other words,
it acts on the subsequence $ \tilde{\bf x} $ of the transmitted sequence $ \bf x $).
We refer to the compound channel as the Noisy Shift Channel with parameters
$ P, K, \mathsf{C}_0 $, or $ \nshift(P; K; \mathsf{C}_0) $ for short.

\begin{theorem}
\label{thm:nshift}
  The zero-error capacity of the Noisy Shift Channel $ \nshift(P; K; \mathsf{C}_0) $
equals $ \log r $, where $ r $ is the unique positive solution to $ x^{K + 1} - 2^{\mathsf{C}_0} x^{K} - 1 = 0 $.
\end{theorem}
\begin{IEEEproof}
  We only give a brief outline of the proof.
A statement analogous to Proposition \ref{thm:CPK} holds in this case too:
if $ \mathcal{C}_\sml{1; K}(n, W) $ is an optimal zero-error code of length $ n $ and
weight $ W $ for the $ \shift(1; K) $, and $ \mathcal{C}^\textsc{n}(W) $ an optimal
zero-error code of length $ W $ for the discrete memoryless channel acting on the
particles, then
\begin{equation}
\label{eq:CN}
  \Big\{ {\bf x} \in \{0, 1, \ldots P\}^n \; : \; \tilde{\bf x} \in \mathcal{C}^\textsc{n}(W) , \; \underline{\bf x} \in \mathcal{C}_\sml{1;K}(n, W) \Big\}
\end{equation}
is an optimal zero-error code of length $ n $ and Hamming weight $ W $ for the
$ \nshift(P; K; \mathsf{C}_0) $.
Its cardinality is $ | \mathcal{C}^\textsc{n}(W) | \cdot | \mathcal{C}_\sml{1;K}(n, W) | $,
and since $ | \mathcal{C}^\textsc{n}(W) | = 2^{\mathsf{C}_0 W + o(n)} $ when $ n \to \infty $,
$ W \sim w n $, further analysis is the same as in the proof of Theorem \ref{thm:cap}
with $ P $ replaced by $ 2^{\mathsf{C}_0} $ (see \eqref{eq:MPKW}).
\end{IEEEproof}

\section{Zero-error capacity of FIFO queues}
\label{sec:queue}

We now turn to the analysis of the $ \queue(P; K; \varphi) $, a channel
introduced as an abstraction of a single-server queue with an infinite buffer.
The proofs rely on the methods used in the previous section for the shift channel.

\subsection{Optimal Codes and the Capacity of the $ \queue $}

As for the shift channel, it is enough to solve the constant-weight case
with $ P = 1 $.
Also, the set of inputs of length $ n $ can again be identified with the simplex
$ \Delta_{n\sml{-W}}^\sml{W} = \big\{ (s_1, \ldots, s_\sml{W}) \in \mathbb{Z}^{W} : 
                                0 \leq s_1 \leq \cdots \leq s_\sml{W} \leq n-W \big\} $.
Before stating the main result of this section, we describe the construction
of optimal codes on a simple example.

\begin{example}
\label{exmpl}
\textnormal{
  Consider the $ \queue(1; 2; \varphi) $, and let $ n = 10 $ and $ W = 2 $.
The set of binary sequences of length $ 10 $ and weight $ 2 $ is represented
as the simplex $ \Delta_8^2 $ in Figure \ref{fig:queue}.
We construct a code by using a procedure analogous to the one used for the
shift channel in \cite[Sec.\ II.B]{kovacevic+popovski}:
List the allowed inputs in the reverse lexicographic order, and in each step
select as a codeword the first sequence available on the list that does not
conflict with previously chosen codewords, i.e., that cannot produce the same
output as one of them.
The resulting code is depicted in Figure \ref{fig:queue1}.
Now observe that we can replace the codewords lying on the right edge of the
simplex with other codewords---$ (0, 0) $ with $ (0, 2) $, $ (3, 3) $ with
$ (3, 5) $, and $ (6, 6) $ with $ (6, 8) $---without affecting the size of
the code and its zero-error property.
Note that the points near the right edge represent the sequences whose $ 1 $'s
are too close so that they can ``push'' each other (think of packets sent in
slots not too far apart, so that processing  one of them may cause the others
to wait in the queue and be further delayed).
The result of this replacement of codewords is the same as if we had first forbidden
the input sequences with $ 1 $'s too close to each other, and then constructed
a code in the same way as for the shift channel;
this is illustrated in Figure \ref{fig:queue2}.
Namely, the effect of the $ \queue(1; K; \varphi) $ on the inputs with $ 1 $'s
separated by at least $ K $ zeros is the same as the effect of the $ \shift(1; K) $
on those inputs---each $ 1 $ is shifted for $ \leq K $ positions to the right.
Finally, notice that expelling the sequences with $ 1 $'s separated by $ < K $
zeros leaves the shape of the space unchanged---it is still a simplex of the same
dimension, only smaller.
}
\myqed
\end{example}

\begin{figure}[h]
\centering
\subfigure[Optimal code obtained by a greedy construction applied on
           binary sequences listed in the reverse lexicographic order.]
{
  \centering
  \includegraphics[width=0.9\columnwidth]{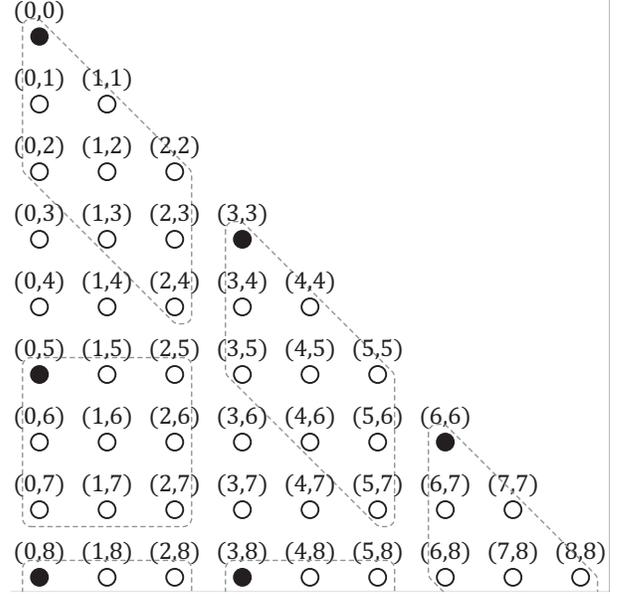}
  \label{fig:queue1}
}
\vspace{1mm}
\subfigure[Optimal code obtained by the same construction as for the $ \shift(1; 2) $,
           after excluding the sequences with $ 1 $'s less than $ 2 $ positions apart
					 (represented by dots in the grey region).]
{
  \centering
  \includegraphics[width=0.9\columnwidth]{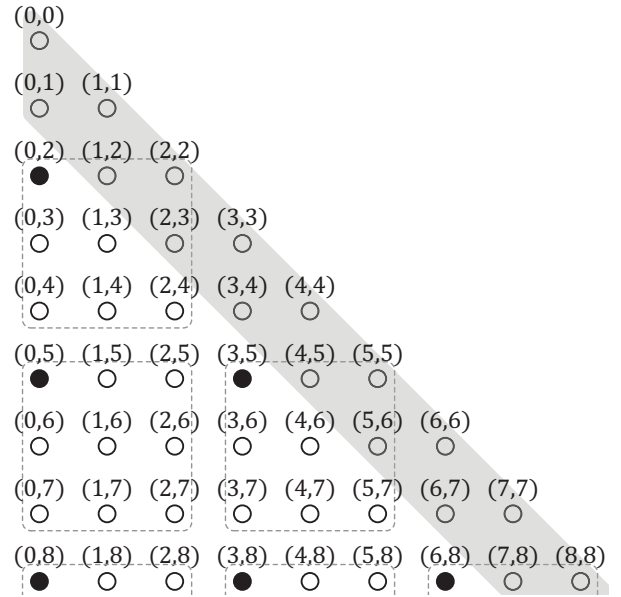}
  \label{fig:queue2}
}
\caption{Zero-error codes of length $ n = 10 $ and weight $ W = 2 $ for the
         $ \queue(1; 2; \varphi) $. Dashed lines illustrate sets of sequences
				 that a given codeword can produce at the output of the $ \queue(1; 2; \varphi) $.}
\label{fig:queue}
\end{figure}%

\begin{theorem}
\label{thm:queue}
  The zero-error capacity of the $ \queue(P; K; \varphi) $ equals
$ \max \Big\{ \frac{\log (P + 1)}{K + 1} ,\;
              \frac{\log P}{\mathbb{E}_\varphi[\kappa] + 1} \Big\} $,
where $ \mathbb{E}_\varphi[\kappa] = \sum_{k=0}^K k \varphi(k) $.
\end{theorem}
\begin{IEEEproof}
  Let $ M_\sml{P;K}^{\textsc q}(n,W) $ denote the size of an optimal zero-error
code of length $ n $ and weight $ W $ for the $ \queue(P; K; \varphi) $ (`Q' in
the superscript stands for `Queue').
The code construction described in the previous example can be used in general:
\begin{inparaenum}[1)]
\item
Start with $ \Delta_{n\sml{-W}}^\sml{W} $,
\item
keep only the sequences for which each of the first $ W - 1 $ $ 1 $'s is followed
by at least $ K $ zeros (at least one such sequence exists if and only if
$ n - W - (W - 1) K \geq 0 $),
\item
in the remaining simplex construct a code in the same way as for the
$ \shift(1; K) $.
\end{inparaenum}
We shall skip the somewhat tedious argument, but it can be shown that this
construction produces an optimal zero-error code when%
\footnote{The greedy construction in the reverse lexicographic order is always optimal,
but it does not necessarily give the same number of codewords as the construction
given by steps 1)--3).
Namely, we have to make sure that the points with which we are replacing the
codewords at the edge of the simplex are themselves in the given simplex, see
Figure \ref{fig:queue}; this is why the stated conditions on  $ n $ are needed.}
$ n \geq W (K + 1) - K $ and $ n \equiv 1\; (\operatorname{mod}\; K+1) $,
e.g., via the adjacency reducing mapping theorem \cite[Thm 3]{shannon}.
Therefore, for $ n \geq W(K+1) - K $, $ n \equiv 1\; (\operatorname{mod}\; K+1) $,
we have $ M_\sml{1;K}^{\textsc q}(n,W) = \big| \Delta_{d}^\sml{W} \big| $,
where $ d = \big\lfloor \frac{n-W-(W-1)K}{K+1} \big\rfloor $, and so
\begin{equation}
\label{eq:MQ}
  M_\sml{1;K}^{\textsc q}(n,W)
     = \binom{ W + \big\lfloor \frac{n+K-W(K+1)}{K+1} \big\rfloor }{ W }  
     = \binom{ \frac{n+K}{K+1} }{ W } .
\end{equation}
For general $ P $, $ M_\sml{P;K}^{\textsc q}(n,W) = P^W \cdot M_\sml{1;K}^{\textsc q}(n,W) $.
The average length of the output sequences is in this case $ L_{\text{av}}(n) \leq n + K $
because consecutive packets are separated by at least $ K $ empty slots by construction
and cannot affect each other's total delay.
From this we get, for $ 0 \leq w < \frac{1}{K+1} $,
\begin{equation}
\label{eq:cwcap2}
 \begin{aligned}
  {\mathcal R}_\sml{P;K}^{\textsc q}(w)
					 &\triangleq  \lim_{n\to\infty} \frac{1}{L_{\text{av}}(n)} \log M_\sml{P;K}^{\textsc q}(n,wn)  \\
					 &=           \frac{1}{K + 1} {\mathcal H}\left( w(K+1) \right)  + w \log P .
 \end{aligned}
\end{equation}
(For the purpose of determining $ {\mathcal R}_\sml{P;K}^{\textsc q}(w) $, it is not
a loss of generality to restrict to lengths $ n \equiv 1\; (\operatorname{mod}\; K+1) $
because one can use zero-padding to satisfy this condition, without affecting the
asymptotic rate of codes and their zero-error property.)
Now consider the case $ w \geq \frac{1}{K+1} $.
For such weights, the construction in the reverse lexicographic order produces
at most a polynomial (in $ n $) number of codewords, e.g., for $ n = W (K + 1) + 1 $
we have $ M_\sml{1;K}^{\textsc q}(n,W) = W + 1 $.
The asymptotic rate will not be reduced if we keep only a single codeword which
minimizes the expected output length, and that is $ 0^W $ ($ W $ packets sent in
the first $ W $ slots).
This will produce $ P^W $ codewords for general $ P $, with the expected output length
of $ L_{\text{av}}(n) = \max\{n, W(\mathbb{E}_\varphi[\kappa] + 1) \} $
(see Example \ref{exmplL} in Section \ref{sec:ze}).
Therefore, for $ \frac{1}{K+1} \leq w \leq 1 $,
\begin{equation}
\label{eq:cwcap3}
  {\mathcal R}_\sml{P;K}^{\textsc q}(w)  =  \frac{ w \log P }{ \max\{1, w(\mathbb{E}_\varphi[\kappa] + 1) \} } .
\end{equation}
Finally, maximizing $ {\mathcal R}_\sml{P;K}^{\textsc q}(w) $ over all $ w $
(see \eqref{eq:cwcap2} and \eqref{eq:cwcap3}) gives the expression for the
zero-error capacity.
\end{IEEEproof}

The capacity-achieving strategy is very simple:
If the capacity equals $ \log (P + 1) / (K + 1) $ it can be achieved by
inserting $ K $ zeros/empty slots after every symbol of the information
sequence written in the alphabet $ \{ 0, 1, \ldots, P \} $, and if it equals
$ \log P / (\mathbb{E}_\varphi[\kappa] + 1) $ the capacity-achieving code
is $ \{ 1, \ldots, P \}^n $.

\subsection{Additional Noise}
\label{sec:queueaddnoise}

Suppose that the packets, apart from being delayed in the queue, experience
other types of impairments as well.
Suppose further that these additional impairments are modeled as a discrete
memoryless channel with input alphabet $ \{1, \ldots, P\} $, with output alphabet
\emph{not} containing the symbol $ 0 $, and with the zero-error capacity equal to
$ \mathsf{C}_0 $ (this channel acts on the packets independently of their passing
through the queue, i.e., it acts on the subsequence $ \tilde{\bf x} $ of the
transmitted sequence $ \bf x $).
We refer to the compound channel as the Noisy $ \queue $ with parameters $ P, K, \mathsf{C}_0 $,
or $ \nqueue(P; K; \mathsf{C}_0) $ for short.

\begin{theorem}
\label{thm:nqueue}
  The zero-error capacity of the $ \nqueue(P; K; \mathsf{C}_0) $ equals
$ \max \Big\{ \frac{\log \left(2^{\mathsf{C}_0} + 1\right)}{K + 1},\;
              \frac{\mathsf{C}_0}{\mathbb{E}_\varphi[\kappa] + 1} \Big\} $.
\end{theorem}
\begin{IEEEproof}
  The proof is analogous to that of Theorem \ref{thm:nshift}---the key is to
focus on the constant-weight case and to observe that the ``effective size'' of
the alphabet of the discrete memoryless channel acting on the packets is $ 2^{{\mathsf{C}_0}} $.
The result is then obtained by replacing $ P $ with $ 2^{{\mathsf{C}_0}} $ in
Theorem \ref{thm:queue}.
\end{IEEEproof}

\section{Zero-error detection}
\label{sec:detect}

In some situations, it is required of the receiver only to \emph{detect}
that a specific kind of error has happened, not necessarily to correct it.
A code $ {\mathcal D}(n) $ is said to be \emph{zero-error-detecting}
for a given channel if it ensures that all possible errors allowed in the
model can be detected, meaning that the receiver can conclude with probability
one whether the transmission was error-free or not.
We shall assume that every input sequence $ \bf x $ can produce itself at
the channel output, i.e., $ {\bf x} \rightsquigarrow {\bf x} $,
because otherwise the detection is trivial (for the $ \shift(P; K_1, K_2) $
this amounts to assuming $ K_1 \leq 0 \leq K_2 $).
If this is the case, then an equivalent way of stating the zero-error-detection
property of a code $ {\mathcal D}(n) $ is that no codeword $ {\bf x} \in {\mathcal D}(n) $
can produce another \emph{codeword} $ \bf y \neq {\bf x} $ at the channel output.
This condition is less stringent compared to the definition of zero-error
code (which will be called zero-error-correcting in this section, to avoid confusion):
two codewords are now allowed to produce the same output $ \bf z $, but as
long as $ \bf z $ itself is not a codeword, the receiver will recognize that
an error has occurred.
The zero-error-detection capacity \cite{gargano, ahlswede2} of a channel is
the $ \limsup $ of the rates of optimal zero-error-detecting codes of length
$ n \to \infty $ for that channel.

\subsection{The Shift Channel}
\label{sec:zedshift}

Unlike in the error-correction case, the channels $ \shift(P; K_1, K_2) $ and
$ \shift(P; K_2 - K_1) $ are not equivalent from the point of view of error-detection,
i.e., the analog of Lemma \ref{thm:noleft} does not hold here.
As an example, consider the code $ \{10000, 00100\} $ which is zero-error-detecting
in the $ \shift(1;-1,1) $, but is not zero-error-detecting in the $ \shift(1;0,2) $,
because in the latter case $ 10000 \rightsquigarrow 00100 $.
However, the analog of Proposition \ref{thm:CPK} holds and enables one to focus on
the case $ P = 1 $.

The following claim describes a relation between zero-error-detecting and zero-error-%
correcting codes for the shift channel.

\begin{proposition}
\label{thm:corr-det}
  Let $ K_1 \leq 0 \leq K_2 $.
\begin{itemize}
\item[(a)]
Every zero-error-detecting code for the $ \shift(P ; K_1, K_2) $ is a
zero-error-correcting code for the $ \shift(P; \min\{ |K_1|, K_2 \} ) $.
\item[(b)]
Every zero-error-correcting code for the $ \shift(P ; \max\{ |K_1|, K_2 \}) $
is a zero-error-detect\-ing code for the $ \shift(P; K_1, K_2 ) $.
\end{itemize}
\end{proposition}

In particular, a code is zero-error-detecting for the $ \shift(P; -K, K ) $ if and
only if it is zero-error-correcting for the $ \shift(P; K ) $.

\begin{IEEEproof}
  Assume w.l.o.g.\ that $ |K_1| \geq K_2 $, and recall the geometric representation
of the code space as described in Section \ref{sec:shift_opt} (constant-weight case,
$ P = 1 $).
Let $ {\bf x} = (x_1, \ldots, x_{\sml W}) $ be a codeword.
That a code is zero-error-detecting for the $ \shift(1; K_1, K_2) $ means that every
hypercube of the form $ \{ {\bf z} : K_1 \leq z_i - x_i \leq K_2 \} $ is such that
it does not contain a codeword other than $ \bf x $.
This, together with the assumption $ |K_1| \geq K_2 $, implies that the hypercubes
$ \{ {\bf z} :  0 \leq z_i - x_i \leq K_2 \} $, formed in this way for every codeword
$ {\bf x} $, are pairwise disjoint, meaning that the code is zero-error-correcting
for the $ \shift(1; K_2) $.
The statement (b) is deduced in a similar way from the geometric interpretation of
the involved notions.
\end{IEEEproof}

Consequently, the zero-error-detection capacity of the $ \shift(P ; K_1, K_2) $ is
lower bounded by the zero-error-correction capacity of the $ \shift(P; \max\{ |K_1|, K_2 \}) $
and upper bounded by the zero-%
error-correction capacity of the $ \shift(P; \min\{ |K_1|, K_2 \} ) $.
We next prove that this upper bound can always be achieved.

\begin{theorem}
  Let $ K_1 \leq 0 \leq K_2 $.
The zero-error-detection capacity of the $ \shift(P ; K_1, K_2) $ is equal
to $ \log s $, where $ s $ is the unique positive real root of the polynomial
$ x^{\min\{ |K_1|, K_2 \}+1} - P x^{\min\{ |K_1|, K_2 \}} - 1 $.
\end{theorem}
\begin{IEEEproof}
  Again, assume that $ |K_1| \geq K_2 $.
As remarked above, Proposition \ref{thm:corr-det}($ a $) implies that the zero-error-%
detection capacity of the $ \shift(P ; K_1, K_2) $ is upper bounded by the zero-error-%
correction capacity of the $ \shift(P; K_2) $, which is precisely $ \log s $ by Theorem
\ref{thm:cap}.
To prove the claim we need to demonstrate that the rate $ \log s $ is achievable,
and this is done by exhibiting a family of codes with the desired properties.
Define
\begin{align}
\label{eq:D1K}
   {\mathcal D}&_{\sml{1; K_2}}^{(a)}(n,W) =  \\ \nonumber
	  &\left\{ {\bf x} \in \Delta_{n\sml{-W}}^\sml{W}\; :\; {\bf x} = {\bf 0} \;\; (\operatorname{mod} \; K_2+1), \;
			                                                       \sum_{i=1}^W x_i = a  \right\} ,
\end{align}
where $ 0 \leq a \leq W (n - W) $ (recall that $ 0 \leq x_i \leq n - W $).
Note that $ {\mathcal D}_{\sml{1; K_2}}^{(a)}(n,W) $ is a subcode of the code
$ {\mathcal C}_\sml{1;K_2}(n,W) $ from \eqref{eq:simplexcode}, obtained as its
intersection with the hyperplane $ \sum_i x_i = a $.
We have
\begin{equation}
\label{eq:CD1K}
  {\mathcal C}_\sml{1;K_2}(n,W) = \bigcup_{a=0}^{W (n - W)} {\mathcal D}_{\sml{1; K_2}}^{(a)}(n,W) ,
\end{equation}
and so, for every $ n $ and $ W $, there is at least one $ a $ for which it holds that
\begin{equation}
  \big| {\mathcal D}_{\sml{1; K_2}}^{(a)}(n,W) \big| \geq \frac{ | {\mathcal C}_\sml{1; K_2}(n,W) | }{ W (n - W) + 1 } .
\end{equation}
Therefore, for $ a $'s chosen in this way, the codes $ {\mathcal D}_{\sml{1; K_2}}^{(a)}(n,W) $
have asymptotically the same rate as the codes $ {\mathcal C}_\sml{1; K_2}(n,W) $, which
is $ \log s $ for $ W \sim w^* n $.\\
It is left to verify that the codes $ {\mathcal D}_{\sml{1; K_2}}^{(a)}(n,W) $
are indeed zero-error-detecting for the $ \shift(1; K_1, K_2) $.
Suppose that a codeword $ {\bf x} \in {\mathcal D}_{\sml{1; K_2}}^{(a)}(n,W) $ was
transmitted and a sequence $ \bf z $ received at the output of the channel.
If $ \sum_i z_i \neq a $, the receiver will easily recognize an error, so suppose
that $ \sum_i z_i = a $.
In this case, if any shifts have occurred in the channel, some of them must have
been shifts to the right and some of them shifts to the left for otherwise we could
not have $ \sum_i x_i \neq \sum_i z_i $.
Suppose that the $ j $'th particle was shifted to the right, $ z_j > x_j $.
Then, since $ x_j = 0 \; (\operatorname{mod} \; K_2+1) $ and $ z_j - x_j \leq K_2 $,
we have $ z_j \neq 0 \; (\operatorname{mod} \; K_2+1) $, so $ \bf z $ cannot
be a codeword.
Therefore, the receiver can detect all errors allowed in the model.
\end{IEEEproof}

Notice that the zero-error-detection capacity of the $ \shift(P; K) $ equals
$ \log(P+1) $ for every $ K $, as if there were no shifts at all.


\subsection{FIFO Queues}
\label{sec:zedqueuet}

In the $ \queue $ model only shifts to the right are possible, which makes the
detection problem very easy (see also the last remark in the previous subsection).

\begin{theorem}
  The zero-error-detection capacity of the $ \queue(P; K; \varphi) $ is equal
to $ \log(P + 1) $.
\end{theorem}
\begin{IEEEproof}
  The value $ \log(P + 1) $ is clearly an upper bound on the capacity because
$ P + 1 $ is the cardinality of the input alphabet, so it is left to prove
achievability.
  The codes
\begin{equation}
\label{eq:Dqueue}
    {\mathcal D}^{(a)}(n,W) = \left\{ {\bf x} \in \Delta_{n\sml{-W}}^\sml{W}\; :\; \sum_{i=1}^W x_i = a  \right\}
\end{equation}
are zero-error-detecting for the $ \queue(1; K; \varphi) $.
Their cardinality, for appropriately chosen $ a $, satisfies
\begin{equation}
  \big| {\mathcal D}^{(a)}(n,W) \big| \geq \frac{ | \Delta_{n\sml{-W}}^\sml{W} | }{ W (n - W) + 1 } .
\end{equation}
Multiplying the above expression by $ P^W $, taking the logarithm, normalizing by
$ n $, and letting $ n \to \infty $, $ W \sim w n $, we get a lower bound on the
asymptotic rate of optimal constant-weight zero-error-detecting codes for the
$ \queue(P; K; \varphi) $ in the form $ {\mathcal H}(w) + w \log P $.
The maximum of this function over $ w \in [0,1] $ is precisely $ \log(P + 1) $.\\
An important point to emphasize here is that the rate was computed by normalizing by
$ n $, and not by $ L_\text{av}(n) $ as in the error-correction case.
The reason is the following: since we are using constant-weight codes, and since
we are only trying to detect the shifts, the receiver can stop looking at the output
after the $ n $'th slot because if some of the packets have been delayed for more
that, it can detect this by counting the received packets in the first $ n $
slots.
The actual rate of the code $ {\mathcal D}^{(a)}(n,W) $ in the error-detection
context is therefore $ \frac{1}{n} \log |{\mathcal D}^{(a)}(n,W)| $.
\end{IEEEproof}

\section{Continuous-time models}
\label{sec:ct}

In this section we introduce and analyze the continuous-time versions of the
shift and queuing channels studied up to this point.
The reasoning is analogous to the discrete-time case so we give only a brief
outline.

Throughout the section we shall assume that the probability distribution of
particle/packet delays is absolutely continuous with respect to the Lebesgue
measure.
This assumption, in particular, ensures the existence of an optimal decoding
rule, i.e., decoding rule that minimizes the error probability.
A code will be called zero-error if its error probability (under an optimal
decoding rule) is equal to zero.

\begin{remark}[Zero-error codes]
\label{rem:ze}
\textnormal{
  Recall (Section \ref{sec:ze}) that, in the discrete-time case, a code is said
to be zero-error if either of the following two equivalent conditions holds:
\begin{inparaenum}[(c1)]
\item
the error probability under optimal decoding is equal to zero, and
\item
no two codewords can produce the same output.
\end{inparaenum}
These two requirements are in general not equivalent in the continuous-time
case---the error probability for a given code can be zero even if two different
codewords can produce the same output, because there are uncountably many possible
outputs.
However, it should be noted that both of these definitions result in the same
value of the zero-error capacity.
The reason we have adopted (c1) as the definition of zero-error codes in this
section is that this convention slightly simplifies the proofs.
}
\myqed
\end{remark}

\subsection{Continuous-Time Shift Channels}
\label{sec:ctshift}

We describe the continuous-time version of the shift channel in the context of
queuing systems.
Suppose that the transmitter can send packets from a $ P $-ary alphabet at
arbitrary instants of time, but with the restriction that any two emissions
are separated by at least $ \tau > 0 $ seconds (think of $ \tau $ as the time
needed to physically transmit a single packet).
Suppose that $ W $ packets were transmitted in a given interval $ T $
(the code ``length'' is now a continuous parameter $ T \in \mathbb{R}_\sml{+} $).
Every such input of duration $ T $ and weight $ W $ is uniquely specified
(for $ P = 1 $) by the sequence of emission times $ (s_1, \ldots, s_\sml{W}) \in \mathbb{R}^{W} $, 
$ 0 \leq s_1 \leq s_2 - \tau \leq s_3 - 2 \tau \leq \cdots \leq s_\sml{W} - (W - 1) \tau \leq T - W \tau $.
Therefore, the set of all inputs of duration $ T $ and weight $ W $ is in a
one-to-one correspondence with the simplex
$ \big\{ (s_1, \ldots, s_\sml{W}) \in \mathbb{R}^{W} : 0 \leq s_1 \leq \cdots \leq s_\sml{W} \leq T/\tau - W \big\} $
(for convenience, we have scaled the emission times with $ \tau $ and subtracted
the vector $ (0, 1, \ldots, W-1) $ from them; this is how the latter representation
was obtained).

We further assume that the $ i $'th packet is delayed in the channel for a random
amount of time $ t_{i,\text{res}} \in [ 0, T_\text{res} ] $, but that reordering of
packets is not possible.
The probability distribution of $ t_{i,\text{res}} $ is assumed to be absolutely continuous
with respect to the Lebesgue measure, with the corresponding density strictly positive
on $ [ 0, T_\text{res} ] $.
In other words, the packets are processed in a FIFO manner and the total time that
any packet spends in the system---the so-called residence time---is bounded by $ T_\text{res} $.
Using the above notation, the set of outputs $ \{{\bf z} : {\bf x} \rightsquigarrow {\bf z}\} $
can be represented as the hypercube of sidelength $ T_\text{res}/\tau $ with $ \bf x $
at its corner (restricted to the simplex), namely
$ \big\{ (z_1, \ldots, z_\sml{W}) \in \mathbb{R}^{W} :
0 \leq z_1 \leq \cdots \leq z_\sml{W} \leq T/\tau - W + T_\text{res}/\tau,\;
0 \leq z_i - x_i \leq T_\text{res}/\tau \big\} $.
We refer to the channel just described as the $ P $-ary Continuous-Time Shift
Channel, $ \ctshift(P;\tau;T_\text{res}) $.

\begin{theorem}
\label{thm:ctshift}
  The zero-error capacity of the $ \ctshift(P;\tau;T_\text{res}) $ equals
$ \frac{1}{\tau} \log v $, where $ v $ is the unique positive solution to
$ x^{\max\{T_\text{res} / \tau, 1\}} - P x^{\max\{T_\text{res} / \tau, 1\} - 1} - 1 = 0 $.
\end{theorem}
\begin{IEEEproof}
As in the discrete case, codewords can be chosen so that the hypercubes
$ \{{\bf z} : {\bf x} \rightsquigarrow {\bf z}\} $ pack the simplex perfectly,
implying that the resulting code is optimal, see Theorem \ref{thm:simplexcode}.
(In the continuous case we allow the decoding regions to overlap, but their
intersection is required to have measure zero, see Remark \ref{rem:ze};
in other words, the hypercubes can touch along their faces only.)
If $ T_\text{res} \geq \tau $, the cardinality of the resulting code will be,
similarly to \eqref{eq:MPKW},
$ P^W \binom{W + \lfloor \frac{T / \tau - W}{T_\text{res} / \tau} \rfloor}{W} $.
The constant-weight zero-error capacity is the limit of the rate of these codes
as $ T \to \infty $ and $ W \sim w T / \tau $.
The zero-error capacity is then obtained by maximizing over $ w \in [0, 1] $
and can be characterized as $ \frac{1}{\tau} \log v $, where $ v $ is the unique
positive solution of $ x^{T_\text{res} / \tau} - P x^{T_\text{res} / \tau - 1} - 1 = 0 $.
If $ T_\text{res} < \tau $, the capacity is trivially $ \frac{1}{\tau} \log(P + 1) $.
\end{IEEEproof}

\subsection{Continuous-Time FIFO Queues}
\label{sec:ctqueue}

Consider now the continuous-time analog of the $ \queue $.
As for the $ \ctshift $, we assume that the transmitter is sending packets from
a $ P $-ary alphabet at arbitrary instants of time, with the restriction that any
two emissions are separated by at least $ \tau > 0 $ seconds.
Further, we assume that the processing time of each packet is a random variable
with distribution absolutely continuous with respect to the Lebesgue measure,
and with the corresponding density $ \varphi(t) $ strictly positive on the interval
$ [0, T_\text{proc}] $.
The service procedure is FIFO.
We refer to this model as the $ P $-ary Continuous-Time Queue with bounded Processing
Time, or $ \ctqueue(P; \tau; T_\text{proc}; \varphi) $ for short.

In a way analogous to the proofs of Theorems \ref{thm:queue} and \ref{thm:ctshift}
we deduce the following result.

\begin{theorem}
\label{thm:ctqueue}
  The zero-error capacity of the $ \ctqueue(P; \tau; T_\text{proc}; \varphi) $ is
$ \max \Big\{ \frac{\log (P + 1)}{\max\{T_\text{proc}, \tau\}},\;
              \frac{\log P}{\max\{\mathbb{E}_\varphi[\kappa], \tau\}} \Big\} $,
where $ \mathbb{E}_\varphi[\kappa] = \int_0^{T_\text{proc}} t \varphi(t) dt $.
\hfill \IEEEQED
\end{theorem}

Notice that the capacity of the $ \ctqueue(P; \tau; T_\text{proc}; \varphi) $ is
independent of $ \tau $ when this parameter is small.
This is an important difference compared to the continuous-time shift channel
discussed in the previous subsection.
For example, when the emission time $ \tau \to 0 $, the zero-error capacity of
the $ \ctshift(P; \tau; T_\text{res}) $ grows to infinity.
This is expected because $ \tau \to 0 $ means that we can send an unbounded number
of packets in any given interval of time, while the delay of each of them is bounded
by a constant $ T_{\text{res}} $.
In the $ \ctqueue $, however, sending more packets also means that the time needed to
receive them will be much longer on average, and the rate in fact remains unchanged.

\section{Concluding remarks}
\label{sec:concl}

Channels with symbol shifts as the dominant type of noise are well-motivated
communication models.
In the present paper two classes of such channels were studied, both in discrete
and continuous time, and a characterization of their zero-error capacity and
zero-error-detection capacity was obtained.
To conclude the paper, we mention two possible extensions of these models as directions
for further work; we believe that these extensions are natural and important in
the context of the mentioned applications.

One of the extensions refers to models that include \emph{deletions} of particles/packets.
One can imagine a queuing system with a finite buffer which drops an incoming packet
whenever the buffer is full, or a molecular communication system in which some of
the particles never arrive at the receiving side.
Another extension are models in which \emph{reordering} of particles/packets is allowed
(see, e.g., \cite{ahlswede, kobayashi, langberg}).
For example, due to properties of most molecular communication systems, it is reasonable
to assume that the order in which the particles arrive at the receiving side is not
necessarily the same as the one in which they were transmitted.
If the particles are identical, then this reordering has no effect on information
transfer and the analysis is the same as for the $ \shift(1; K) $ \cite{kovacevic+popovski};
however, the case $ P > 1 $ seems to be much more difficult and the corresponding
analysis would require different methods than the ones used here.

Precisely defining and analyzing models similar to those studied in this paper, but
which also include deletions and/or out-of-order arrival of packets, is an interesting
problem for future investigation.

\section*{Acknowledgment}

The authors would like to thank the three anonymous reviewers for their comments
which greatly improved the presentation of this work.


\appendix
\label{sec:functions}

We list here several properties of the zero-error capacity of the $ \shift(P;K) $ and
related quantities, regarded as functions of the channel parameters.
The parameters $ P $ and $ K $ are assumed to be integers taking values $ P \geq 1 $
and $ K \geq 0 $.
Whenever the behavior of a function with respect to one variable is discussed,
it is understood that the remaining variables/parameters are kept fixed.

\begin{proposition}
\label{thm:rlogr}
The function $ r $, defined by $ r^{K+1} - P r^K - 1 = 0 $, $ r > 0 $, is
\begin{itemize}
\item[(a)]
Continuous, monotonically decreasing, and convex in $ K $, with
$ r_{\rvert \sml{K=0}} = P+1 $ and $ \lim_{\sml{K} \to \infty} r = P $;
\item[(b)]
Continuous, monotonically increasing, and convex in $ P $, with
$ \lim_{\sml{P} \to \infty} \frac{r}{P} = 1 $.
\end{itemize}
\pagebreak
  The function $ {\mathcal R}_\sml{P;K}^* = \log r $ is
\begin{itemize}
\item[(c)]
Continuous, monotonically decreasing, and convex in $ K $;
\item[(d)]
Continuous, monotonically increasing in $ P $, and concave over $ P \geq 2 $.
\end{itemize}
\end{proposition}
\begin{IEEEproof}
  The functions $ r $ and $ \log r $ are well-defined for arbitrary real
(not necessarily integer) $ K $ and $ P $ in the specified ranges.
The claim is obtained by differentiating them, e.g.,
\begin{equation}
\label{eq:difrK}
  \dot{r}_\sml{K} = \frac{ - r \ln r }{ r^K( (K + 1) r - K P ) }
                  = \frac{ - (r - P) r \ln r }{ (K + 1)(r - P) + P } ,
\end{equation}
and verifying the sign of the derivatives.
\end{IEEEproof}


The function $ {\mathcal R}_\sml{P;K}(w) $ is even easier to analyze since
it is explicit, see \eqref{eq:cwcap}.

\begin{proposition}
  The function $ {\mathcal R}_\sml{P;K}(w) $ is
\begin{itemize}
\item[(a)]
Continuous, monotonically decreasing, and convex in $ K $;
\item[(b)]
Continuous, monotonically increasing, and concave in $ P $;
\item[(c)]
Continuous and concave in $ w \in [0,1] $.
\hfill \IEEEQED
\end{itemize}
\end{proposition}

The values/limits of $ {\mathcal R}_\sml{P;K}(w) $ at $ K = 0 $, $ K \to \infty $,
$ P = 1 $, $ P \to \infty $, $ w = 0 $ and $ w = 1 $, can be found directly from
\eqref{eq:cwcap}.

Finally, we state several properties of the weight (the fraction of occupied cells)
which optimizes the rate of a constant-weight code.

\begin{proposition}
\label{thm:w}
  Define $ w^* = \argmax_{ w \in [0,1] } {\mathcal R}_\sml{P;K}(w) $.
The function $ w^* $ has the following properties:
\begin{itemize}
\item[(a)]
$ w^* = \frac{ P }{ (K+1)(r-P) + P } < 1 $;
\item[(b)]
It is continuous, monotonically increasing in $ P $, and concave over $ P \geq 2 $;
\item[(c)]
For $ P \geq 2 $, it is monotonically increasing in $ K $, with
$ \lim_{\sml{K} \to \infty} w^* = 1 $;\\
For $ P = 1 $, it is monotonically decreasing in $ K $, with
$ \lim_{\sml{K} \to \infty} w^* = 0 $.
\hfill \IEEEQED
\end{itemize}
\end{proposition}
\begin{IEEEproof}
  Equating the derivative of $ {\mathcal R}_\sml{P;K}(w) $ with zero we get that
$ w^* $ is the solution of
\begin{equation}
\label{eq:wopt}
 \frac{w^* K + 1}{w^* (K + 1)} \cdot \left( \frac{1 - w^*}{w^* K + 1} \right)^\frac{1}{K+1} \cdot P = 1 .
\end{equation}
Letting $ h = P(w^* K + 1) / (w^* (K + 1)) $, \eqref{eq:wopt} becomes
$ h^{K+1} - P h^K - 1 = 0 $, which means that $ h = r $.
This proves {\em (a)}. 
{\em (b)} is shown by calculating the derivatives of $ w^* $ from {\em (a)}.
To prove {\em (c)}, it is enough to demonstrate that the function $ (K+1)(r-P) $ is
monotonically decreasing to $ 0 $ for $ P \geq 2 $, and monotonically increasing
to $ \infty $ for $ P = 1 $, which can again be shown by analyzing its derivative.
\end{IEEEproof}



\begin{thebibliography}{99}

\bibitem{ahlswede}
   R. Ahlswede and A. H. Kaspi,
   ``Optimal Coding Strategies for Certain Permuting Channels,"
   \emph{IEEE Trans. Inf. Theory}, vol. 33, no. 3, pp. 310--314, May 1987.
\bibitem{ahlswede2}
   R. Ahlswede, N. Cai, and Z. Zhang,
   ``Erasure, List, and Detection Zero-Error Capacities for Low Noise and a Relation to Identification,''
   \emph{IEEE Trans. Inf. Theory}, vol. 42, no. 1, pp. 55--62, Jan. 1996.
\bibitem{anantharam}
   V. Anantharam and S. Verd\'u,
   ``Bits Through Queues,"
   \emph{IEEE Trans. Inf. Theory}, vol. 42, no. 1, pp. 4--18, Jan. 1996.
\bibitem{bedekar}
   A. S. Bedekar and M. Azizo\~{g}lu,
   ``The Information-Theoretic Capacity of Discrete-Time Queues,"
   \emph{IEEE Trans. Inf. Theory}, vol. 44, no. 2, pp. 446--461, Mar. 1998.
\bibitem{engelberg+keren}
   S. Engelberg and O. Keren,
   ``Reliable Communications Across Parallel Asynchronous Channels With Arbitrary Skews,''
   \emph{IEEE Trans. Inf. Theory}, vol. 63, no. 2, pp. 1120--1129, Feb. 2017.
\bibitem{farsad2}
   N. Farsad, H. B. Yilmaz, A. Eckford, C.-B. Chae, and W. Guo,
	 ``A Comprehensive Survey of Recent Advancements in Molecular Communication,''
	 \emph{IEEE Commun. Surveys Tuts.}, vol. 18, no. 3, pp. 1887--1919, 2016.
\bibitem{gargano}
   L. Gargano, J. K\"orner, and U. Vaccaro,
	 ``Qualitative Independence and {S}perner Problems for Directed Graphs,''
	 \emph{J. Combin. Theory Ser. A}, vol. 61, no. 2, pp. 173--192, Nov. 1992.
\bibitem{immink}
   K. A. S. Immink,
   ``Runlength-Limited Sequences,"
   \emph{Proc. IEEE}, vol. 78, no. 11, pp. 1745--1759, Nov. 1990.
\bibitem{kobayashi}
   K. Kobayashi,
   ``Combinatorial Structure and Capacity of the Permuting Relay Channel,"
   \emph{IEEE Trans. Inf. Theory}, vol. 33, no. 6, pp. 813--826, Nov. 1987.
\bibitem{korner}
   J. K\"orner and A. Orlitsky,
   ``Zero-Error Information Theory,"
   \emph{IEEE Trans. Inf. Theory}, vol. 44, no. 6, pp. 2207--2229, Oct. 1998.
\bibitem{kovacevic}
   M. Kova\v{c}evi\'c,
   ``A Note on Parallel Asynchronous Channels With Arbitrary Skews,''
   \emph{IEEE Trans. Inf. Theory}, to appear.
\bibitem{kovacevic+popovski}
   M. Kova\v{c}evi\'c and P. Popovski,
   ``Zero-Error Capacity of a Class of Timing Channels,''
   \emph{IEEE Trans. Inf. Theory}, vol. 60, no. 11, pp. 6796--6800, Nov. 2014.
\bibitem{krachkovsky}
   V. Yu. Krachkovsky,
   ``Bounds on the Zero-Error Capacity of the Input-Constrained Bit-Shift Channel,"
   \emph{IEEE Trans. Inf. Theory}, vol. 40, no. 4, pp. 1240--1244, Jul. 1994.
\bibitem{langberg}
   M. Langberg, M. Schwartz, and E. Yaakobi,
   ``Coding for the $ \ell_\infty $-Limited Permutation Channel,''
   in \emph{Proc. 2015 IEEE Int. Symp. Inf. Theory (ISIT)}, pp. 1936--1940, Hong Kong, Jun. 2015.
\bibitem{nakano}
   T. Nakano, A. W. Eckford, and T. Haraguchi,
	 \emph{Molecular Communication},
	 Cambridge University Press, 2013.
\bibitem{sellke}
   S. H. Sellke, C.-C. Wang, N. Shroff, and S. Bagchi,
	 ``Capacity Bounds on Timing Channels with Bounded Service Times,"
	 in \emph{Proc. 2007 IEEE Int. Symp. Inf. Theory (ISIT)}, pp. 981--985, Nice, France, Jun. 2007.
\bibitem{shamai}
   S. Shamai (Shitz) and E. Zehavi,
   ``Bounds on the Capacity of the Bit-Shift Magnetic Recording Channel,"
   \emph{IEEE Trans. Inf. Theory}, vol. 37, no. 3, pp. 863--872, May 1991.
\bibitem{shannon}
   C. E. Shannon,
   ``The Zero Error Capacity of a Noisy Channel,"
   \emph{IRE Trans. Inf. Theory}, vol. 2, no. 3, pp. 8--19, Sep. 1956.
\bibitem{tan}
   V. Y. F. Tan,
	 ``Asymptotic Estimates in Information Theory with Non-Vanishing Error Probabilities,"
	 \emph{Foundations and Trends in Communications and Information Theory}, vol. 11, nos. 1-2, pp. 1--184, 2014.
\bibitem{thomas}
   J. A. Thomas,
	 ``On the {S}hannon Capacity of Discrete Time Queues,"
	 in \emph{Proc. 1997 IEEE Int. Symp. Inf. Theory (ISIT)}, p. 333, Ulm, Germany, Jun./Jul. 1997.
\bibitem{wilf}
   H. S. Wilf,
   ``{P}erron-{F}robenius Theory and the Zeros of Polynomials,"
   \emph{Proc. Amer. Math. Soc.}, vol. 12, no. 2, pp. 247--250, Apr. 1961.
\bibitem{wilfbook}
   H. S. Wilf,
   \emph{Mathematics for the Physical Sciences},
   Dover Publications, Inc., 1978.

\end{thebibliography}
\end{document}